\documentclass[a4paper,UKenglish,cleveref, autoref, thm-restate,authorcolumns]{lipics-v2019}

\usepackage{microtype}
\usepackage{pict2e, xcolor, graphics, graphicx, amssymb}
\usepackage{times}

\usepackage{url}
\usepackage{tikz}
\usepackage{graphicx}
\usetikzlibrary{positioning}
\usepackage{amsmath}
\usepackage{color}
\usepackage{amsfonts}%
\usepackage{amssymb}%
\usepackage{graphicx}

\usepackage{amsfonts}
\usepackage{amsthm}
\usepackage[noend]{algorithm}
\usepackage[noend]{algorithmic}

\renewcommand\proofname{\bf Proof}

\newcommand{\St}{\mathrm{St}}

\newcommand{\su}{\mathrm{succ}}

\usepackage{comment}

\title{Interval Query Problem on Cube-free Median Graphs} 

\author{Soh Kumabe}{The University of Tokyo}{soh\_kumabe@mist.i.u-tokyo.ac.jp}{}{}

\authorrunning{S.Kumabe} 

\Copyright{Soh Kumabe} 



\category{} 

\relatedversion{} 

\supplement{}

\nolinenumbers 

\EventEditors{John Q. Open and Joan R. Access}
\EventNoEds{2}
\EventLongTitle{42nd Conference on Very Important Topics (CVIT 2016)}
\EventShortTitle{CVIT 2016}
\EventAcronym{CVIT}
\EventYear{2016}
\EventDate{December 24--27, 2016}
\EventLocation{Little Whinging, United Kingdom}
\EventLogo{}
\SeriesVolume{42}
\ArticleNo{23}

\begin{document}

\maketitle

\begin{abstract}
In this paper, we introduce the \emph{interval query problem} on cube-free median graphs. Let $G$ be a cube-free median graph and $\mathcal{S}$ be a commutative semigroup. For each vertex $v$ in $G$, we are given an element $p(v)$ in $\mathcal{S}$. For each query, we are given two vertices $u,v$ in $G$ and asked to calculate the sum of $p(z)$ over all vertices $z$ belonging to a $u-v$ shortest path. This is a common generalization of range query problems on trees and grids. In this paper, we provide an algorithm to answer each interval query in $O(\log^2 n)$ time. The required data structure is constructed in $O(n\log^3 n)$ time and $O(n\log^2 n)$ space. To obtain our algorithm, we introduce a new technique, named the \emph{staircases decomposition}, to decompose an interval of cube-free median graphs into simpler substructures.
\end{abstract}

\section{Introduction}

The \emph{range query problem}~\cite{gabow1984scaling} is one of the most fundamental problems in the literature on data structures, particularly for string algorithms~\cite{gusfield1997algorithms}. Let $f$ be a function defined on arrays. In the range query problem, we are given an array $P=(p(1),\dots, p(n))$ of $n$ elements and a \emph{range query} defined by two integers $i,j$ with $1\leq i\leq j\leq n$. For each query $(i,j)$, we are asked to return the value $f((p(i),\dots, p(j)))$.
The main interest of this problem is the case where $f$ is defined via a \emph{semigroup operator}~\cite{yao1982space}. Let $\mathcal{S}$ be a semigroup with operator $\oplus$, and let $P$ consist of elements in $\mathcal{S}$. Then, the function $f$ is defined as $f((p(i),\dots, p(j)))=p(i)\oplus \dots \oplus p(j)$. Typical examples of semigroup operators are sum, max, and min. The fundamental result~\cite{yao1982space,yao1985complexity} is that for any constant integer $k$, a range query can be answered in $O(\alpha_k(n))$ time, where $\alpha_k$ is a slow-growing function related to the inverse of the Ackermann function. The required data structure is constructed in linear time and space. \emph{Range minimum query problem}, i.e., $\oplus=\min$, is one of the well-studied problems in the literature, and it admits a constant-time algorithm with a data structure constructed in linear time and space~\cite{alstrup2002nearest, bender2000lca,bender2005lowest,gabow1984scaling,harel1984fast}.

This problem is generalized into trees and grids. In these settings, we are given a tree/grid $G$ and an element $p(v)$ for each vertex of $G$. As a query, given two vertices $u,v$ in $G$, we are asked to calculate the sum~\footnote{In this paper, for simplicity, we represent the semigroup operation by the terms of summation; that is, we denote $a\oplus a'$ by the word \emph{sum} of $a$ and $a'$ for $a,a'\in \mathcal{S}$.} of the elements assigned at the vertices on a $u-v$ shortest path. In particular, we are asked to calculate the sum of the elements on the unique $u-v$ path for trees and the axis-parallel rectangle with corners $(u,v)$ on its diagonal for grids. For constant dimensional grids, an almost-constant time algorithm~\cite{chazelle1989computing} with linear space on semigroup operators and a constant-time algorithm for range minimum query is known~\cite{yuan2010data}. For range query problem on trees, an almost-constant time algorithm~\cite{chazelle1987computing} with linear space is known on semigroup operators; see~\cite{brodal2011path} for further survey on the problem on trees, particularly for dynamic version.

In this paper, we introduce a common generalization of the two above mentioned cases, named \emph{interval query problem} on \emph{median graphs}. 
Let $G=(V(G),E(G))$ be a connected graph with $n$ vertices. For two vertices $u,v\in V(G)$, let the \emph{interval} $I[u,v]$ be the set of vertices belonging to a $u-v$ shortest path, where the length of a path is defined by the number of its edges. The graph $G$ is called a \emph{median graph} if for all $u,v,w\in V(G)$, $I[u,v]\cap I[v,w]\cap I[w,u]$ is a singleton~\cite{avann1961metric,birkhoff1947ternary,nebesky1971median}. The median graph $G$ is said to be \emph{cube-free} if $G$ does not contain a cube as an induced subgraph. Trees and grids are examples of cube-free median graphs. 
In our problem, we are given a median graph $G$ and an element $p(v)$ of a commutative semigroup $\mathcal{S}$ for each vertex $v$ of $G$. As a query, given two vertices $u,v$ in $G$, we are asked to calculate $p(I[u,v])$~\footnote{For a vertex subset $X$, we denote the sum of $p(z)$ over all $z\in X$ by $p(X)$.}. The interval query problem on cube-free median graphs is a common generalization of the range query problems on trees and grids.

In this paper, we provide an algorithm to the interval query problem on cube-free median graphs. The main result here is presented as follows:

\begin{theorem}
There is an algorithm to answer interval queries on cube-free median graphs in $O(\log^2 n)$ time. The required data structure is constructed in $O(n\log^3 n)$ time and $O(n\log^2 n)$ space, where $n$ is the number of vertices in a given cube-free median graph.
\end{theorem}

The time complexity of answering a query matches the complexity for the two-dimensional \emph{range tree}~\cite{lueker1978data} in the \emph{orthogonal range query problem}, without acceleration via \emph{fractional cascading}~\cite{chazelle1986fractional}. 

To obtain the algorithm, we introduce a new technique, named the \emph{staircases decomposition}. This technique provides a new method to decompose an interval of cube-free median graphs into a constant number of smaller intervals. Most of the candidates of the smaller intervals, which we refer to as \emph{staircases}, are well-structured, and an efficient algorithm to answer the interval queries can be constructed. The rest are not necessarily staircases; however, each of them are one of the $O(n\log n)$ candidates, and we can precalculate all the answers of the interval queries on these intervals.

Designing fast algorithms for median graphs is a recently emerging topic. The \emph{distance labeling scheme}~\cite{peleg2000proximity} is a type of data structure that is defined by the encoder and decoder pair. The encoder receives a graph and assigns a label for each vertex, whereas the decoder receives two labels and computes the distance of the two vertices with these labels. For cube-free median graphs, there is a distance labeling scheme that assigns labels with $O(\log^3 n)$ bits for each vertex~\cite{chepoi2019distance}. Very recently, a linear-time algorithm to find the \emph{median} of median graphs was built~\cite{beneteau2019medians}. This paper continues with this line of research and utilizes some of the techniques presented in these previous studies.

Various applications can be considered in the interval query problem on median graphs. The solution space of a 2-SAT formula forms a median graph, where two solutions are adjacent if one of them can be obtained by negating a set of pairwise dependent variables of the other~\cite{bandelt2008metric,mulder1979median,schaefer1978complexity}. For two solutions $u$ and $v$, the interval $I[u,v]$ corresponds to the set of the solutions $x$, such that for each truth variable, if the same truth value is assigned in $u$ and $v$, so does $x$. Suppose we can answer the interval queries to calculate sum (resp. min) in polylogarithmic time with a data structure of subquadratic time and space. Then, if we have the list of all feasible solutions of the given 2-SAT formula, we can calculate the number (resp. minimum weight) of these solutions in polynomial time of the number of variables for each query, without precalculating the answers for all possible queries. Note that, there is a polynomial-delay algorithm to enumerate all solutions to the given 2-SAT formula~\cite{creignou1997generating}. Therefore, if the number of the feasible solutions (and thus the number of vertices in the corresponding median graph) is small, we can efficiently list them. In social choice theory, the structure of median graphs naturally arises as a generalization of \emph{single-crossing preferences}~\cite{clearwater2015generalizing,demange2012majority} and every closed Condorcet domains admits the structure of a median graph~\cite{puppe2019condorcet}. For two preferences $u$ and $v$, the voters with their preferences in interval $I[u,v]$ prefer candidate $x$ to candidate $y$ whenever both $u$ and $v$ prefer $x$ to $y$. Therefore, using interval query, we can count the number of voters $w$ such that for all pairs of candidates, at least one of $u$ and $v$ has the same preference order as $w$ between these candidates.
Although these structures are not necessarily cube-free, we hope that our result will be the first and important step toward obtaining fast algorithms for these problems.

\subsection{Algorithm Overview}

Here we give high-level intuition to our algorithm. More detailed outline is given in Section~\ref{sec:outline}.

Let $G$ be a cube-free median graph. The first idea for our algorithm is to decompose $G$ recursively. We recursively divide $V(G)$ into some parts, called \emph{fibers}. Roughly speaking, a fiber is a set of the vertices located on the similar direction from the special vertex $m$ (see figure~(a)). Each fiber induces a cube-free median graph and, if we take $m$ properly, has at most $|V(G)|/2$ vertices; there are at most $O(\log n)$ recursion steps.

Let $u,v$ be vertices of $G$. Consider calculating $p(I[u,v])$. If $u$ and $v$ are in the same fiber, we calculate it recursively.
Otherwise, we can show that $I[u,v]$ intersects with only a constant number of fibers and the intersections are intervals with one end on the boundary of the fiber (Section~\ref{sec:decomposition}, see Figure~(h)). Thus, it is sufficient to construct an algorithm on such intervals.

To do this, we further decompose such an interval into more well-structured intervals, using our main technique named \emph{staircases decomposition} (Section~\ref{sec:stairs}, see Figure~(d),~(e)~and~(f)). Roughly speaking, we decompose the interval into at most two structured substructures names \emph{staircases} (figure~(b)) and a special interval of $O(n)$ candidates. For special intervals $I$, we just use the precalculated $p(V(I))$. For staircases $L$, we construct an algorithm to calculate $p(V(L))$ in $O(\log^2 n)$ time (Section~\ref{sec:process}), using the fact that the boundary of the fiber is actually a tree~\cite{chepoi2019distance}. We decompose this tree into paths by heavy-light decomposition and build segment trees to answer the queries.


\section{Basic Tools for Cube-Free Median Graphs and Trees}\label{sec:prelim}

In this section, we introduce basic facts about cube-free median graphs and trees.

Let $G$ be a connected, undirected, finite graph. We denote the vertex set of $G$ by $V(G)$. For two vertices $u$ and $v$ in $G$, we write $u\sim v$ if $u$ and $v$ are adjacent. For two vertices $u$ and $v$ of $G$, the \emph{distance} $d(u,v)$ between them is the minimum number of edges on a path connecting $u$ and $v$, and the \emph{interval} $I[u,v]$ is the set of vertices $w$ which satisfies $d(u,v)=d(u,w)+d(w,v)$. The graph $G$ is a \emph{median graph} if for any three vertices $u,v,w$, $I[u,v]\cap I[v,w]\cap I[w,u]$ contains exactly one vertex, called \emph{median} of $u,v$ and $w$. Median graphs are bipartite and do not contain $K_{2,3}$ as a subgraph.
A median graph is \emph{cube-free} if it does not contain a (three-dimensional) cube graph as an induced subgraph. The followings hold.

\begin{lemma}[{\rm \cite{chepoi2013shortest}}]\label{lem:subgrid}
Any interval in a cube-free median graph induces an isometric subgraph of a two-dimensional grid.
\end{lemma}

\begin{lemma}[{\rm \cite{chepoi2019distance}}]\label{lem:quad}
Let $u,v,w_1,w_2$ be four pairwise distinct vertices of a median graph such that $v\sim w_1, v\sim w_2$ and $d(u,v)-1=d(u,w_1)=d(u,w_2)$. Then, there is unique vertex $z$ with $w_1\sim z,w_2\sim z$ and $d(u,z)=d(u,v)-2$.
\end{lemma}

From now on, let $G$ be a cube-free median graph with $n$ vertices. Let $X$ be a subset of $V(G)$. For vertex $z\in V(G)$ and $x\in X$, $x$ is the \emph{gate} of $z$ in $X$ if for all $w\in X$, $x\in I[z,w]$. The gate of $z$ in $X$ is unique (if it exists) because it is the unique vertex in $X$ that minimizes the distance from $z$. $X$ is \emph{gated} if all vertices $z\in V(G)$ have a gate in $X$. The following equivalence result is known.

\begin{lemma}[{\rm \cite{chepoi1989classification,chepoi2019distance}}]
Let $X$ be a vertex subset of the median graph $G$. Then, following three conditions are equivalent.
\begin{description}
  \item[(a)] $X$ is gated.
  \item[(b)] $X$ is \emph{convex}, i.e., $I[u,v]\subseteq X$ for all $u,v\in X$. 
  \item[(c)] $X$ induces a connected subgraph and $X$ is \emph{locally convex}, i.e., $I[u,v]\subseteq X$ for all $u,v\in X$ with $d(u,v)=2$.
\end{description}
\end{lemma}

An induced subgraph of $G$ is \emph{gated} (resp. \emph{convex}, \emph{locally convex}) if its vertex set is gated (resp. convex, locally convex).
The intersection of two convex subsets is convex. Any interval of median graphs are convex.

For a convex subset $X$ and a vertex $x\in X$, the \emph{fiber} $F_X(x)$ of $x$ \emph{with respect to} $X$ is the set of vertices in $G$ whose gate in $X$ is $x$. Two fibers $F_X(x), F_X(y)$ are \emph{neighboring} if there are vertices $x'\in F_X(x)$ and $y'\in F_X(y)$ such that $x'\sim y'$, which is equivalent to $x\sim y$~\cite{chepoi2019distance}.
Fibers for all $x\in X$ define a partition of $V(G)$. For two adjacent vertices $x,y\in X$, the \emph{boundary} $T_X(x,y)$ of $F_X(x)$ \emph{relative to} $F_X(y)$ is the set of the vertices which have a neighbor in $F_X(y)$. $T_X(x,y)$ and $T_X(y,x)$ are isomorphic. A vertex in $T_X(x,y)$ has a unique neighbor in $T_X(y,x)$, which is the corresponding vertex under that isomorphism.
For vertex $x\in X$, a \emph{total boundary} $T_X(x)$ of $F_X(x)$ is the union of all $T_X(x,y)$ for $y\in X$ with $x\sim y$.
The subgraph $H$ is \emph{isometric} in $G$ if for all $u,v\in V(H)$, there is a path in $H$ with length $d(u,v)$.
A rooted tree has \emph{gated branches} if any of its root-leaf path is convex. The next lemma exploits the structures of the boundaries of fibers of cube-free median graphs. 

\begin{lemma}\label{lem:gtree}\label{lem:samedist}\label{lem:uniqueneighbor}\label{lem:neighboringpath}
Let $X$ be a convex vertex subset of cube-free median graph $G$. Let $x,y\in X$ and assume $x\sim y$. Then, the followings hold.\\
{\it (i)}~{\rm (\cite{chepoi2019distance})} $T_X(x,y)$ induces a tree, which is convex.\\
{\it (ii)}~{\rm (\cite{chepoi2019distance})} $T_X(x)$ induces a tree with gated branches, which is isometric in $G$.
\end{lemma}

The following is folklore in a literature of median graphs. A proof is in Appendix~\ref{app:proof}.

\begin{lemma}[{\rm folklore}]\label{lem:convexstar}
Let $X$ be a convex vertex set of a median graph and let $Y$ be a convex subset of $X$. For $x\in X$, let $F(x)$ be the fiber of $x$ with respect to $X$. Then, $\bigcup_{y\in Y}F(y)$ is convex.
\end{lemma}

Let $T$ be a tree with gated branches. For a vertex $v\in V(G)$ and $w\in T$, $w$ is an \emph{imprint} of $v$ if $I[v,w]\cap T=\{w\}$. If $T$ is convex, the imprint is equal to the gate and therefore unique. Even if it is not the case, we can state following.

\begin{lemma}\label{lem:equaldist}
Let $T$ be a tree with gated branches rooted at $r$. Let $u\in V(G)$. Then, the following statements hold.\\
{\it (i)}~{\rm (\cite{chepoi2019distance})} There are at most two imprints of $u$ in $T$.\\
{\it (ii)} Assume $u$ has two distinct imprints $w^1,w^2$ in $T$. Then, $w^1,w^2\in I[r,u]$.
\end{lemma}
\begin{proof}
We prove (ii). From symmetry, we only prove $w^1\in I[r,u]$. Let $P_1$ be the root-leaf path of $T$ that contains $w^1$. Then, $P_1$ is convex and therefore $d(r,u)=d(r,w^1)+d(w^1,u)$.
\end{proof}

\begin{lemma}
Let $T$ be a tree with gated branches and $w\in V(T)$. Then, the set of vertices with an imprint $w$ in $T$ is convex. 
\end{lemma}
\begin{proof}
Assume the contrary. Then, there are distinct vertices $z_1,z_2,z_3$ with $z_1\sim z_2\sim z_3$, such that $z_1$ and $z_3$ have an imprint $w$ but $z_2$ doesn't. We have $d(w,z_2)=d(w,z_1)+1$; otherwise, $d(w,z_2)=d(w,z_1)-1$ because of bipartiteness of $G$ holds and in this case, $I[w,z_2]\subseteq I[w,z_1]$ holds and $z_2$ has an imprint $w$. By the same reason we have $d(w,z_2)=d(w,z_3)+1$. From definition of the imprint, there is a $z_2-w$ shortest path that contains a vertex of $T$ other than $w$. Let $z_4$ be the neighbor of $z_2$ in this shortest path. Then, $z_4$ does not have an imprint $w$ and especially, $z_1\neq z_4\neq z_3$. Now we have $d(w,z_4)+1=d(w,z_1)+1=d(w,z_3)+1=d(w,z_2)$ and obtain three squares that all two intersect at an edge from Lemma~\ref{lem:quad}, which contradicts Lemma~\ref{lem:subgrid}.
\end{proof}


For a vertex $m\in V(G)$, the \emph{star} $\St(m)$ of $m$ is the set of vertices $x\in V(G)$ such that there is an edge or a square that contains both $m$ and $x$. $\St(m)$ is convex. 
The vertex $m\in V(G)$ is a \emph{median} of $G$ if it minimizes the sum of distances to all vertices in $G$. The following holds.

\begin{lemma}[{\rm \cite{chepoi2019distance}}]\label{lem:half}
All the fibers of $\St(m)$ of a median graph contains at most $\frac{n}{2}$ vertices.
\end{lemma}

For a rooted tree $T$ that is rooted at $r$, a vertex $u\in V(T)$ is an \emph{ancestor} of $v$ and $v$ is a \emph{descendant} of $u$ if there is a path from $u$ to $v$, only going toward the leaves. The vertex subset $X$ is a \emph{column} of $T$ if for any two vertices $x,y$ in $X$, $x$ is either an ancestor or a descendant of $y$. The vertex $t$ is the \emph{lowest common ancestor}~\cite{harel1984fast} of $u$ and $v$ if $t$ is an ancestor of both $u$ and $v$ that minimizes the distance between $u$ and $t$ (or equivalently, $v$ and $t$) in $T$.
There is a data structure that is constructed in linear time and space such that, given two vertices on $T$, it returns the lowest common ancestor of them in constant time~\cite{bender2005lowest}. 
$u$ is a \emph{parent} of $v$ and $v$ is a \emph{child} of $u$ if $u$ is an ancestor of $v$ and $u\sim v$. Let $X\subseteq V(T)$ and $u\in V(T)$. The \emph{nearest ancestor} of $u$ in $X$ on $T$ is the vertex $v\in X$ such that $v$ is an ancestor of $u$ and minimizes $d(u,v)$.

Let $T$ be a rooted tree rooted at $r$. For a vertex $v\in V(T)$, let $T_v$ be the subtree of $T$ rooted at $v$. An edge $(u,v)$ in $G$ such that $u$ is the parent of $v$ is a \emph{heavy-edge} if $|V(T_u)|\leq 2|V(T_v)|$ and a \emph{light-edge} otherwise. Each vertex has at most one child such that the edge between them is a heavy-edge. The \emph{heavy-path} is a maximal path that only contains heavy-edges. The \emph{heavy-light decomposition} is the decomposition of $T$ into heavy-paths. Note that, there is at most $O(\log n)$ light-edges on any root-leaf path on $T$.

\section{Outline and Organization}\label{sec:outline}

Here we roughly describe our algorithm using the notions in Section~\ref{sec:prelim}.
Let $G$ be a cube-free median graph. Let $m$ be a median of $G$, $\St(m)$ be the star of $m$, and for each $x\in \St(m)$, let $F(x)$ be the fiber of $x$ in $\St(m)$ (see figure~(a)). Let $u,v$ be vertices of $G$.

Consider calculating $p(I[u,v])$. If $u$ and $v$ are in the same fiber $F(x)$ of $\St(m)$, we calculate the answer by using the algorithm on $F(x)$, which is recursively defined. Lemma~\ref{lem:half} ensures that the recursion depth is at most $O(\log n)$.
Otherwise, we can show that $I[u,v]$ intersects with only a constant number of fibers, and for each fiber $F(x)$ that intersects $I[u,v]$, $I[u,v]\cap F(x)$ can be represented as $I[u_x,v_x]$ for some vertices $u_x, v_x\in F(x)$ such that $v_x$ is on the total boundary of $F(x)$. Thus, it is sufficient to construct an algorithm to answer the query on the interval, such that one of the ends is on the total boundary of $F(x)$.

To do this, we introduce a technique to decompose intervals, which we name the \emph{staircases decomposition}. Let $T$ be a tree with gated branches and assume $u\in V(G)$ and $v\in V(T)$. We partition an interval $I[u,v]$ into an interval $I$ and at most two special structures, which we name a \emph{staircases} (figure~(b)), which we describe in Section~\ref{sec:stairs}. Such a decomposition can be calculated in $O(\log n)$ time with appropriate preprocessing. Here, we can take $I$ as one of the $O(n)$ candidates of intervals. We just precalculate and store the value $p(I)$ for each candidate, and recall it when we answer the queries.

Now we just need an algorithm to calculate the value $p(V(L))$ quickly for a staircases $L$. Let $P$ be a root-leaf path of $T$. The segment trees can answer the staircases queries whose base is a subpath of $P$ in $O(\log n)$ time. To answer the general queries, we use a heavy-light decomposition of $T$.

The rest of the paper is organized as follows. In Section~\ref{sec:stairs}, we introduce the staircases decomposition of the intervals with one end on the tree with gated branches. In Section~\ref{sec:process}, we construct an algorithm and a data structure for the interval queries for the same cases. In Section~\ref{sec:decomposition}, we prove that we can decompose a given interval into constant number of intervals with one of the ends on the total boundaries of the fibers of $\St(m)$. This technique can also be applied to the query that asks the median of given three vertices. Some detailed parts in these sections are found in Appendix~\ref{sec:finding}. Finally, in Appendix~\ref{sec:construction}, we give an algorithm to construct our data structure efficiently.

\section{The Staircases Decomposition of the Intervals with One End on the Boundary}\label{sec:stairs}

\begin{figure}
\begin{tabular}{c}
    \begin{minipage}{0.33\hsize}
    \centering
          \includegraphics[scale=1.0]{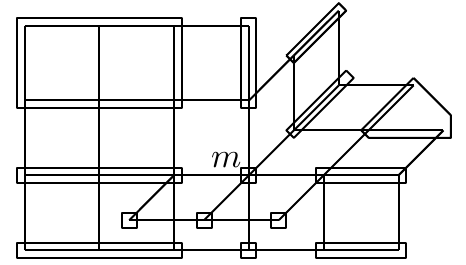}\\
          \flushleft{(a) a median graph and\\ its decomposition \\into fibers of $\St(m)$.}
    \end{minipage}
    \begin{minipage}{0.33\hsize}
    \centering
          \includegraphics[scale=1.0]{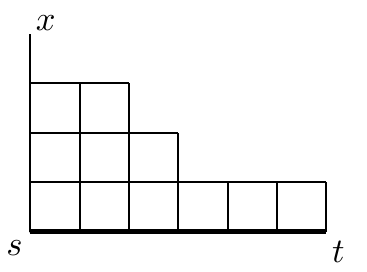}\\
          \flushleft{(b) staircases with top $x$\\ with base starts at $s$\\ and ends at $t$.}
    \end{minipage}
    \begin{minipage}{0.33\hsize}
    \centering
          \includegraphics[scale=0.8]{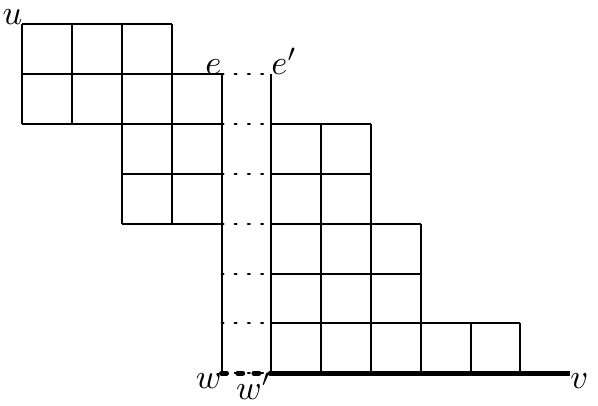}\\
          \flushleft{(c) decomposition of\\ $I[u,v]$ into an interval $I[u,w]$\\ and staircases $I[e',v]$.\\ The bold line represents $P$.}
    \end{minipage}
\end{tabular}
\end{figure}

Let $T$ be a tree with gated branches. In this section, we introduce a technique, \emph{staircases decomposition}, to decompose an interval $I[u,v]$ such that $v$ is on $T$. 

Let $P=(s=w_0,\dots, w_k=t)$ be a convex path. For a vertex $x$ with gate $s$ in $P$, the interval $I[x,t]$ induces \emph{staircases} if for all $i=0,\dots, k$, the set of vertices in $I[x,t]$ with gate $w_i$ in $P$ induces a path. $P$ is the \emph{base} of $L$ and the vertex $x$ is the \emph{top} of $L$. The base \emph{starts} at $s$ and \emph{ends} at $t$ (see figure~(b)). Our staircases decomposition decomposes $I[u,v]$ into an interval and at most two staircases such that their bases are columns of $T$.

\subsection{The case with One End on a Convex Path}\label{sec:path}

Here we investigate the structure of an interval such that one of the endpoints is on a convex path $P$.
Consider an interval $I[u,v]$ such that $v$ is on $P$. Let $w$ be the gate of $u$ in $P$. The purpose here is to prove that $I[u,v]$ can be decomposed into the disjoint union of an interval $I[u,w]$ and a staircases (see Figure~(c)), if $w\neq v$. We assume $w\neq v$ because otherwise we have no need of decomposition. Let $w'$ be the neighbor of $w$ in $P$ between $w$ and $v$. We take the isometric embedding of $I[u,v]$ into a two-dimensional grid (see Lemma~\ref{lem:subgrid}). We naively introduce a $xy$-coordinate system with $w=(0,0)$, $w'=(1,0)$, $u=(x_u,y_u)$ with $y_u\geq 0$, and $v=(x_v,y_v)$ with $y_v\leq 0$. Now, we can state the following.

\begin{lemma}
If a vertex $z=(x_z,y_z)$ on $I[u,v]\cap V(P)$ is not on the $x$-axis, there is no vertex other than $z$ in $I[u,v]$ with gate $z$ in $P$.
\end{lemma}
\begin{proof}
Assume the contrary and let $z'=(x_{z'},y_{z'})$ be a vertex in $I[u,v]$ with gate $z$ in $P$.
Because of the isometricity, $x_z>0$ and $y_z<0$ holds. Since $z\in I[w,z']$, we have $x_{z'}\geq x_z$. Since $z'\in I[u,v]$, $x_v\geq x_{z'}$ holds. Therefore, we can take a vertex $z''$ in $P$ with $x$-coordinate $x_{z'}$, but it means $z'\in I[w,z'']$ and contradicts to the convexity of $P$.
\end{proof}

Since such $z$ does not affect the possibility of decomposition (we can just add such vertices at the end of the staircases), we can assume that $v=(x_v,0)$ for $x_v>0$. Moreover, from convexity, we have that all vertices in $I[u,v]$ have non-negative $y$-coordinate. Thus, $I[u,v]\setminus I[u,w]$ is the set of vertices with positive $x$-coordinate and forms staircases (see figure~(c)), which is the desired result.

To build an algorithm to calculate $p(I[u,v])$ as the sum of $p(I[u,w])$ and $p(I[u,v]\setminus I[u,w])$, we should identify the top $e'$ of the staircases. Instead of direct identification, we rather identify the unique neighbor of it in $I[u,w]$, named the \emph{entrance} $e$ of the staircases: The top $e'$ can be determined as the neighbor of $e$ with gate $w'$ on $P$. Here, we have that $e$ is the gate of $u$ in the boundary of $F(w)$ with respect to $F(w')$, where $F(w)$ (resp. $F(w')$) is the fiber of $w$
(resp. $w'$) with respect to $P$. Indeed, this gate should be in $I[u,w]$ from the definition of the gate and $e$ is the only candidate for it.
We can calculate $e$ in $O(\log n)$ time by working on the appropriate data structure on total boundary of the fiber of $w$ with respect to $P$. We discuss this algorithm in Appendix~\ref{sec:finding}.

\subsection{Single Imprint}\label{sec:single}

\begin{figure}
\begin{tabular}{c}
    \begin{minipage}{0.33\hsize}
    \centering
          \includegraphics[scale=0.8]{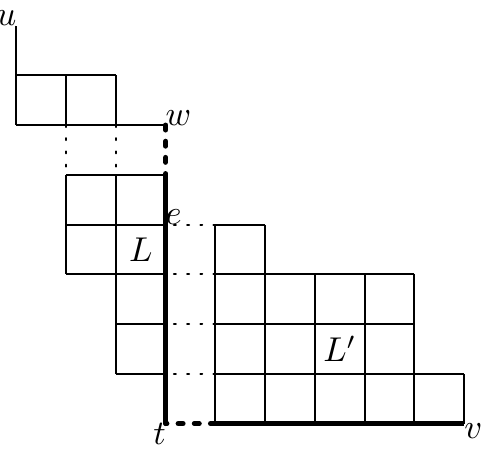}\\
          \flushleft{(d) staircases decomposition of $I[u,v]$.(single imprint, first case)}
    \end{minipage}
    \begin{minipage}{0.33\hsize}
    \centering
          \includegraphics[scale=0.8]{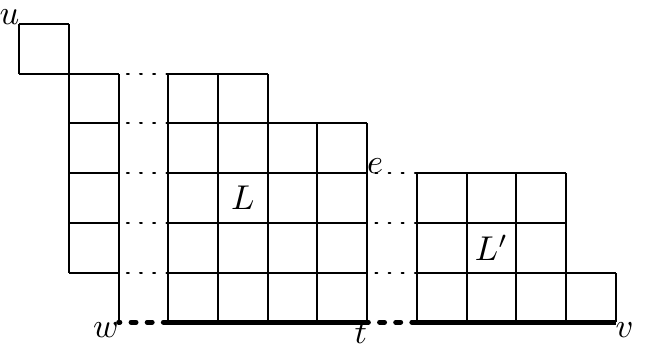}\\
          \flushleft{(e) staircases decomposition of $I[u,v]$.(single imprint, second case)}
    \end{minipage}
    \begin{minipage}{0.33\hsize}
    \centering
          \includegraphics[scale=0.8]{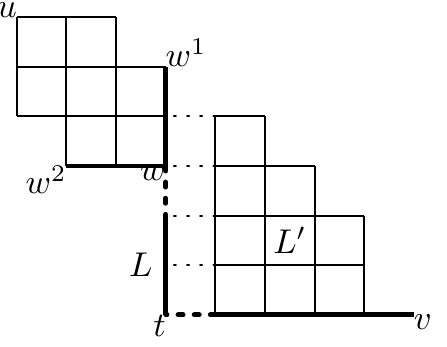}\\
          \flushleft{(f) staircases decomposition of $I[u,v]$.(double imprints)}
    \end{minipage}
\end{tabular}
\end{figure}

Here we give the staircases decomposition of the interval $I[u,v]$, where $v$ is on a tree $T$ with gated branches, rooted at $r$. First, we treat the case that there is exactly one imprint $w$ of $u$ in $T$ in $I[u,v]$. Let $t$ be a lowest common ancestor of $w$ and $v$ in $T$. Note that, $t$ might coincide with $w$ or $v$. Let $P$ (resp. $P'$) be the root-leaf path of $T$ that contains $w$ (resp. $v$).

Since $P'$ is convex, we can decompose $I[u,v]$ into a staircases $L'$ with base on $P'$ and an interval $I[u,t]$. Since $P$ is convex, we can further decompose the interval $I[u,t]$ into a staircases $L$ with base on $P$ and an interval $I[u,w]$. Now, for fixed $T$, $I[u,w]$ is one of the $O(n)$ candidates of the intervals, because it is specified only by a vertex $u$ and one of at most two imprints of $u$ on $T$. This is the staircases decomposition we obtain here.

To bound the size of the data structure we construct in Section~\ref{sec:process}, we should ensure that staircases $L$ and $L'$ contains only vertices with an imprint on $P$ and $P'$, respectively. Let $B_L$ (resp. $B_{L'}$) be the base of $L$ (resp. $L'$). We prove the following.

\begin{lemma}\label{lem:ppintersect}\label{lem:oneimp}
The following statements hold.\\
{\it (i)} $I[u,v]$ contains no vertices in $T$ other than the vertices on the $w-v$ path on $T$.\\
{\it (ii)} For a vertex $z$ in $L'$, the gate of $z$ in $P'$ is an imprint of $z$ in $T$.\\ {\it (iii)} For a vertex $z$ in $L$, the gate of $z$ in $P$ is an imprint of $z$ in $T$.
\end{lemma}
\begin{proof}
{\it (i)} Let $z\in I[u,v]\cap V(T)$. Then, $d(u,v)=d(u,z)+d(z,v)$ holds. Since $w$ is the unique imprint of $u$ in $T$, $d(u,z)=d(u,w)+d(w,z)$ and $d(u,v)=d(u,w)+d(w,v)$ holds. Therefore $d(w,v)=d(w,z)+d(z,v)$ and it means $z$ is on the unique path between $w$ and $v$ on $T$. 
{\it (ii)} Let $w'_z$ be the gate of $z$ in $P'$. We prove $I[z,w'_z]\cap V(T)=\{w'_z\}$. Assume $x\in (I[z,w'_z]\cap V(T))\setminus \{w'_z\}$. From {\it (i)}, $x$ is on $w-t$ path. 
From isometricity of $T$, $t\in I[x,w'_z]\subseteq I[z,w'_z]$ holds and it contradicts the definition of $w'_z$. 
{\it (iii)} Similar to {\it (ii)}.
\end{proof}

We should also make algorithms to identify the top of the staircases $L$ and $L'$. The top of $L$ can be found by applying the discussion in previous subsection by precalculating the entrances for all possible patterns of $u$ and $w$, because the start of the base of $L$ is uniquely determined as a parent of $w$, independent of $v$. However, we cannot apply it to find the top of $L'$, because the start of the base of $L'$ is a child of $t$, not a parent.
Instead, we calculate the top of $L'$ by case-analysis of the positional relation of the staircases. Intuitively, we divide cases by the angle formed by $B_L$ and $B_{L'}$. We have essentially two cases\footnote{To explain all cases by these two, we take $T$ as the maximal tree with gated branches that contains the fiber we consider, rather than the fiber itself.} to tract, which this angle is $\pi/2$ (Figure~(d)) or $\pi$ (Figure~(e)) (we formally define these cases and prove that they cover all cases in Appendix~\ref{sec:finding}). In the case in Figure~(d), the entrance $e$ of $L'$ can be found on $B_L$. In the case in Figure~(e), $e$ can be found on the total boundary of the vertex set with imprint $t$. In both case, by appropriate data structure given in Appendix~\ref{sec:finding}, we can find the entrance in $O(\log n)$ time.

\subsection{Double Imprints}\label{sec:double}

Here we consider the staircases decomposition for the case that there are two imprints $w^1,w^2$ of $u$ in $T$ in $I[u,v]$. Let $w$ be the lowest common ancestor of $w^1$ and $w^2$ in $T$. From (ii) of Lemma~\ref{lem:equaldist} and isometricity of $T$, $d(u,w^1)+d(w^1,w)=d(u,r)-d(w,r)=d(u,w^2)+d(w^2,w)$ holds and particularly we have $w^1,w^2\in I[u,w]$. From isometricity of $T$, we have $w\in I[w_1,w_2]\subseteq I[u,v]$. Let $t$ be the lowest common ancestor of $w$ and $v$. Then, from isometricity of $T$, we have $t\in I[w,v]\subseteq I[u,v]$. Note that, the lowest common ancestor of $v$ and $w^1$ (resp. $w^2$) is also $t$, because otherwise we have $w\not \in I[u,v]$. Let $P$ (resp. $P'$) be any root-leaf path of $T$ that contains $w$ (resp. $v$).

Since $P'$ is convex, we can decompose $I[u,v]$ into a staircases $L'$ with base on $P'$ and an interval $I[u,t]$. Since the subpath of $P$ between $r$ and $w$ is convex, we can further decompose the interval $I[u,t]$ into a staircases $L$ with base on $P$ and an interval $I[u,w]$ (actually, we can prove that $L$ is a line). Now, for fixed $T$, $I[u,w]$ is one of the $O(n)$ candidates of the intervals, because $w$ is specified only by a vertex $u$, as the lowest common ancestor of two imprints of $u$ in $T$. This is the staircases decomposition we obtain here.

Let $B_L$ (resp. $B_{L'}$) be the base of $L$ (resp. $L'$). From the same reason as the case with a single imprint, we prove the following lemma. The proof is similar to the proof of Lemma~\ref{lem:oneimp}.

\begin{lemma}\label{lem:pppintersect}\label{lem:twoimp}
The following statements hold.\\
(i) $I[u,v]$ contains no vertices in $T$ other than vertices in $w^1-v$ and $w^2-v$ path on $T$.\\
(ii) For a vertex $z$ in $L'$, the gate of $z$ in $P'$ is an imprint of $z$ in $T$.\\ (iii) For a vertex $z$ in $L$, the gate of $z$ in $P$ is an imprint of $z$ in $T$.
\end{lemma}
\begin{proof}
{\it (i)} Let $z\in I[u,v]\cap V(T)$. Then, $d(u,v)=d(u,z)+d(z,v)$ holds. Let $w^i$ be the imprint of $u$ in $T$ with $d(u,z)=d(u,w^i)+d(w^i,z)$. Then, $d(u,v)=d(u,w^i)+d(w^i,v)$ holds. Therefore $d(w^i,v)=d(w^i,z)+d(z,v)$ and it means $z$ is on the unique path between $w^i$ and $v$ on $T$.
{\it (ii)} Let $w'_z$ be the gate of $z$ in $P'$. We prove $I[z,w'_z]\cap V(T)=\{w'_z\}$. Assume $x\in (I[z,w'_z]\cap V(T))\setminus \{w'_z\}$. From {\it (i)}, $x$ is on $t-w^1$ or $t-w^2$ path. 
From isometricity of $T$, $t\in I[x,w'_z]\subseteq I[z,w'_z]$ holds and it contradicts the definition of $w'_z$. 
{\it (iii)} Similar to {\it (ii)}.
\end{proof}

We should also provide a way to identify the top of the staircases $L'$. We have only one case to tract, shown in Figure~(f), which we can find the entrance on $w^1-t$ or $w^2-t$ path on $T$ (we formally define the case in Appendix~\ref{sec:finding}). We can find it in $O(\log n)$ time in the algorithm in Appendix~\ref{sec:finding}.

\section{Query Processiing of the Case with One End on the Tree with Gated Branches}\label{sec:process}

In this section, we construct an algorithm and a data structure that answers the queries with one of the endpoints on the tree with gated branches. That part is the core of our algorithm.

\subsection{Query Processing for Maximal Staircases with Base on Convex Path}

Here we construct an algorithm and a data structure for the staircases whose base is contained in a convex path $P$. For simplicity, we assume that $P$ contains $2^q$ vertices for some integer $q$. We do not lose generality by this restriction because we can safely attach dummy vertices at the end of $P$. Let $P=(w_0,\dots, w_{2^q-1})$. Our data structure uses a segment tree defined on $P$. The information of the vertices with base $w_i$ in $P$ are stored by linking to $w_i$.

It is convenient to consider the \emph{direction} of $P$, as if $P$ is directed from $w_0$ to $w_{2^q-1}$. The \emph{reverse} $\bar{P}$ of $P$ is the same path as $P$ as an undirected path but has different direction, i.e., $\bar{P}=(w_{2^q-1},\dots, w_0)$. We represent the path between $w_x$ and $w_y$ on $P$ by $P[x,y]$. 

Let us formally define the queries to answer here. A query is represented by three vertices $x,w_a,w_b$ such that the gate of $x$ on $P$ is $w_a$, and asks to answer the value $p(L(x,w_a,w_b))$, where $L(x,w_a,w_b)$ represents the staircases with top $x$ and base starts at $w_a$ and ends at $w_b$. We construct two data structures, the first one treats the case $a\leq b$ and the second one treats the case $a>b$. The second data structure is just obtained by building the first data structure on the reverse of $P$, therefore we can assume that for all queries, $w_a\leq w_b$ holds.

For $i=0,\dots, 2^q-1$, let $F_i$ be the fiber of $w_i$ with respect to $P$. For $i=0,\dots, 2^q-2$ and $z\in F_i$, the \emph{successor} $\su_P(z)$ of $z$ is the gate of $z$ in $F_{i+1}$ (see Figure~(g)). Intuitively, $\su_P(z)$ represents the next step of $z$ in the staircases with base in $P$; more precisely, for $a<i<b$, if $F_i \cap V(L(x,w_a,w_b))$ induces $z-w_i$ path, $F_{i+1}\cap V(L(x,w_a,w_b))$ induces $\su_P(z)-w_{i+1}$ path.

Here we construct a complete binary tree, which is referred to as \emph{segment tree}, to answer the queries.
For each $d=0,\dots, q$ and for each $i=0,1,\dots, 2^{q-d}-1$, we prepare a node that corresponds to $P[i\times 2^d, (i+1)\times 2^d-1]$. For each node $v$ that corresponds to $P[l,r]$ and for each $z\in F_l$, we store the vertex $s(z,l,r)=\su_P^{r-l}(z)$ and the value $S(z,l,r)=p(L(z,w_l,w_r))=p(I[\su_P^{0}(z),w_{l}])\oplus \dots \oplus p(I[\su_P^{r-l}(z),w_{r}])$, where the $\su_P^k(z)$ is recursively defined by $\su_P^0(z)=z$ and $\su_P^{k+1}(z)=\su_P(\su_P^k(z))$ for all $0\leq k$.

The Algorithm~\ref{alg:stairs} calculates $p(L(x,w_a,w_b))$. We call the procedure $\text{StaircasesQuery}_P(0,2^q-1,a,b,x)$ to calculate it, and the algorithm returns the pair of the vertex $\su_P^{b-a+1}(x)$ and the value $p(L(x,w_a,w_b))$. The time complexity is $O(q)=O(\log n)$.

\begin{algorithm}
\caption{$\text{StaircasesQuery}_P(l,r,a,b,x)$}
\begin{algorithmic}[1]\label{alg:stairs}
\IF{$[l,r]\subseteq [a,b]$}
    \RETURN $(s(x,l,r), S(x,l,r))$
\ENDIF
\STATE $med\leftarrow \lfloor \frac{l+r}{2}\rfloor$
\IF{$b\leq med$}
    \RETURN $\text{StaircasesQuery}_P(l,med,a,b,x)$
\ENDIF
\IF{$med<a$}
    \RETURN $\text{StaircasesQuery}_P(med+1,r,a,b,x)$
\ENDIF
\STATE $(x',S_1)\leftarrow \text{StaircasesQuery}_P(l,med,a,b,x)$
\STATE $(x'',S_2)\leftarrow \text{StaircasesQuery}_P(med+1,r,a,b,\su_P(x'))$
\RETURN $(x'',S_1\oplus S_2)$
\end{algorithmic}
\end{algorithm}

This data structure is constructed as in Algorithm~\ref{alg:stairsconstruction}. The correctness is clear and the time complexity is $O(nq)\leq O(n\log n)$, assuming that we know the vertex $\su_P(x)$ and the value $p(I[x,w_i])$ for all $i=0,\dots, 2^q-1$ and $x\in F_i$. The size of the data structure is clearly $O(nq)\leq O(n\log n)$. We give algorithms to calculate $\su_P(x)$ in Appendix~\ref{sec:finding} and $p(I[x,w_i])$ in Appendix~\ref{sec:construction}.

\begin{algorithm}
\caption{Construction of the Data Structure for Staircases with Base on Convex Path}
\begin{algorithmic}[1]\label{alg:stairsconstruction}
\REQUIRE A cube-free median graph $G$, a convex path $P=(w_0,\dots, w_{2^q-1})$ 
\FOR{$i=0,\dots, 2^q-1$}
    \FORALL{$x\in F_i$}
        \STATE $s(x,i,i)\leftarrow x$
        \STATE $S(x,i,i)\leftarrow p(L(x,w_i,w_i))=p(I[x,w_i])$
    \ENDFOR
\ENDFOR
\FOR{$d=q-1,\dots, 0$}
    \FOR{$i=0,\dots, 2^{q-d}-1$}
        \STATE $a\leftarrow i\times 2^d,b\leftarrow (i+\frac{1}{2})\times 2^d,c\leftarrow (i+1)\times 2^d$
        \FORALL{$x\in F_i$}
            \STATE $s(x,a,c-1)\leftarrow s(\su_P(s,a,b-1),b,c-1)$
            \STATE $S(x,a,c-1)\leftarrow S(x,a,b-1)\oplus S(\su_P(s(x,a,b-1)),b,c-1)$
        \ENDFOR
    \ENDFOR
\ENDFOR
\end{algorithmic}
\end{algorithm}

\begin{figure}
\begin{tabular}{c}
    \begin{minipage}{0.5\hsize}
    \centering
          \includegraphics[scale=1.0]{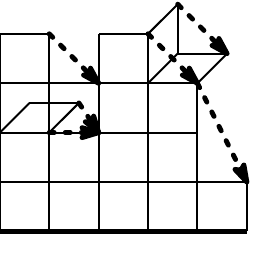}\\
          \hspace{1.2cm} (g) the arrows go from $v$ to $\su_P(v)$.\\ The bold line represents $P$.
    \end{minipage}
    \begin{minipage}{0.5\hsize}
    \centering
          \includegraphics[scale=1.0]{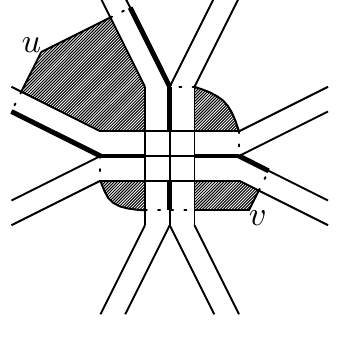}\\
          \hspace{1.2cm} (h) decomposition of $I[u,v]$.
    \end{minipage}
\end{tabular}
\end{figure}

\subsection{Query Processing for Staircases with Base on the Tree with Gated Branches}

Let $T$ be a tree with gated branches. Here we construct an algorithm and a data structure for the staircases whose base is a column of $T$. The simplest idea is to prepare the data structure discussed in the previous subsection for all root-leaf paths on $T$, but in this case the total size of the data structure can be as bad as $O(n^2\log n)$. To reduce the size, we instead prepare the above data structure on every heavy-path of heavy-light decomposition of $T$.

For a vertex $w\in V(T)$, let $F(w)$ be the set of vertices with an imprint $w$.
For an edge $(w,w')$ of $T$ and a vertex $z\in F(w)$, we denote $\su_{w,w'}(z)$ by the gate of $z$ in $F(w')$. For the staircases $L$ whose base starts at $w_1$ and ends at $w_2$ such that $w_1,w,w',w_2$ are located on some column of $T$ in this order, if $F(w)\cap V(L)$ induces $z-w$ path, $F(w')\cap V(L)$ is $\su_{w,w'}(z)-w'$ path.

Let $P$ be a heavy-path of $T$. Let $V_P$ be the set of vertices that has an imprint in $P$. We build a data structure discussed in the previous subsection on the graph induced by $V_P$ together with the convex path $P$; Lemma~\ref{lem:oneimp} and Lemma~\ref{lem:twoimp} ensures that, for any staircases $L$ we want to treat, all the vertices in $L$ has an imprint in the base of $L$.
We can calculate the answer for the queries by Algorithm~\ref{alg:stairs2}, where the vertices in a heavy-path is represented as $P=(w_{P,0},\dots, w_{P,w^{q_P}})$.

\begin{algorithm}
\caption{$\text{StaircasesQuery}(u,w,v)$}
\begin{algorithmic}[1]\label{alg:stairs2}
\REQUIRE $w,v\in V(T)$, $u\in V(G)$ such that $w$ and $v$ are on the same column of $T$ and $u\in F(w)$
\STATE Let $Q$ be the $w-v$ path on $T$ and $P_1,\dots, P_k$ be the list of heavy-paths that contains vertices in $Q$, in the same order appearing in $Q$
\STATE Let $P_1\cap Q=(w=w_{P_1,s_1},\dots,w_{P_1,t_1}),P_2\cap Q=(w_{P_2,s_2},\dots, w_{P_2,t_2}),\dots,P_k\cap Q=(w_{P_k,s_k},\dots, w_{P_k,t_k}=v)$
\STATE $(x,S)\leftarrow \text{StaircasesQuery}_{P_1}(0,2^{q_{P_1}}-1,s_1,t_1,u)$
\FOR{$i=2,\dots,k$}
    \STATE $(x',S')\leftarrow \text{StaircasesQuery}_{P_1}(0,2^{q_{P_i}}-1,s_i,t_i,\su_{w_{P_{i-1},t_{i-1}},w_{P_i,s_i}}(x))$
    \STATE $x\leftarrow x', S\leftarrow S\oplus S'$
\ENDFOR
\RETURN $(x,S)$
\end{algorithmic}
\end{algorithm}

The correctness of the algorithm is clear. The size of the data structure is bounded by $O(n\log n)$, because the size of the data structure on a heavy-path $P$ is bounded by $O(|V_P|\log |V_P|)$ and each vertex is in $V_P$ for at most two heavy-paths $P$. 
We should make an algorithm to calculate the successor efficiently. We describe an algorithm that works in $O(\log n)$ time in Appendix~\ref{sec:finding}.


Now, the time complexity of Algorithm~\ref{alg:stairs2} is $O(\log^2 n)$ because $k$ in the algorithm is at most $O(\log n)$. 

\subsection{Putting them Together}

Here we summarize our work on the interval query problem with one end on the tree with gated branches. In Section~\ref{sec:stairs}, for the fixed tree $T$ with gated branches, we have seen that any interval with one end on $T$ can be decomposed to at most two staircases (say, $L$ and $L'$, for instance we allow any of them to be empty) and a special interval $I$ that is one of $O(n)$ candidates. As we roughly described in Section~\ref{sec:stairs}, such decomposition can be calculated in $O(\log n)$ time (See Appendix~\ref{sec:finding} for details). 

Now we consider calculating the answer as $p(L)+p(L')+p(I)$.
$p(L)$ and $p(L')$ can be calculated in $O(\log^2 n)$ time by above algorithm. Furthermore, $p(I)$ is precalculated in the construction of our data structure and we can take this value in constant time. Therefore we can answer the interval query in the case with one end on the tree with gated branches in $O(\log^2 n)$ time. We summarize our algorithm in Algorithm~\ref{alg:stairsall} in Appendix.

Here we describe how $p(I)$ can be precalculated. Recall that, we construct our data structure recursively on each fibers. Therefore, after constructing the smaller data structure on each fiber, we can calculate the value $p(I)$ in $O(\log^2 n)$ time by using an interval query on them to complete construction. This is the bottleneck part of our construction algorithm, along with $O(\log n)$ recursion steps. Note that, this procedure can be implemented during preprocessing because there are only $O(n)$ candidates of $I$. When answering to the queries, we do not need to use the smaller data structure; we have only to refer these precalculated values.

\section{Decomposing Intervals into intervals with One End on the Boundary}\label{sec:decomposition}

In this section, we consider decomposing an interval with both ends in different fibers into smaller intervals with one end on boundaries (see Figure~(h)). Specifically, we bound the number of such fibers by $9$. Let $m$ be the median of $G$. For $x\in \St(m)$, let $F(x)$ be the fiber of $x$ with respect to $\St(m)$. For $v\in V(G)$, let $r(v)$ be the vertex in $\St(m)$ that is nearest from $v$. From definition of fibers, $v\in F(r(v))$ holds. 

First, we prove that the intersection of an interval and a fiber is indeed an interval. The following lemma holds.

\begin{lemma}\label{lem:fiberintersection}
Let $u,v$ be vertices and let $x\in \St(m)$. Let $g_u,g_v$ be the gate of $u,v$ in $F(x)$, respectively. Then, $I[u,v]\cap F(x)$ coincides with $I[g_u,g_v]$ if it is nonempty.
\end{lemma}
\begin{proof}
Assume $z\in I[u,v]\cap F(x)$. From the definition of the gate, there is a $u-z$ (resp. $v-z$) shortest path that passes through $g_u$ (resp. $g_v$). Therefore there is a $u-v$ shortest path that passes through $u,g_u,z,g_v,v$ in this order, which means $z\in I[g_u,g_v]$.
Converse direction is clear from $I[g_u,g_v]\subseteq I[u,v]$, which is from the definition of the gate.
\end{proof}

Note that, unless $r(u)=r(v)$, one of the gates of $u$ or $v$ in $F(x)$ is on the total boundary of $F(x)$. Therefore, to obtain the desired structural result, we just need to bound the number of fibers with non-empty intersection with $I[u,v]$. We use the following lemma from~\cite{chepoi2019distance}.

\begin{lemma}[{\rm \cite{chepoi2019distance}}]
Let $u,v$ be vertices with $r(u)\neq r(v)$. Then, one of the $m\in I[u,v]$, $r(u)\sim r(v)$, or $d(m,r(u))=d(m,r(v))=d(r(u),r(v))=2$ holds.
\end{lemma}

Assume $m\in I[u,v]$. Then, $I[u,v]\cap F(x)\neq \emptyset$ means $x\in I[u,v]$. Therefore the number of such fibers $F(x)$ is same as the number of vertices in $I[u,v]\cap \St(m)$. Now, from the fact that $I[u,v]$ has a grid structure (see Lemma~\ref{lem:subgrid}) and $\St(m)$ consists of the vertices in an edge or a square that contains $m$, we have that $|I[u,v]\cap \St(m)|\leq 9$.

If $r(u)\sim r(v)$, from Lemma~\ref{lem:convexstar}, we have $I[u,v]\subseteq F(r(u))\cup F(r(v))$. If $d(r(u),r(v))=2$, let $w$ be the unique common neighbor of $r(u)$ and $r(v)$. Then, from Lemma~\ref{lem:convexstar}, we have $I[u,v]\subseteq F(r(u))\cup F(w)\cup F(r(v))$. Therefore, in all cases, the number of fibers with nonempty intersection with $I[u,v]$ is bounded by $9$.

In all of these cases, we can list the fibers $F(x)$ with nonempty intersection with the given interval $I[u,v]$; it is the set of the fibers of the vertices in $I[r(u),r(v)]$ because $\bigcup_{x\in I[r(u),r(v)]}F(x)$ is convex, and, we can list them efficiently by using the list of all squares in $G$. Now it is sufficient to give a way to calculate the gate of $u$ and $v$ in each of these fibers for our algorithm. We give the algorithm in Appendix~\ref{sec:finding}.

Above technique can also be applied for the following query. We are given three vertices $v_1,v_2,v_3$ in a cube-free median graph $G$ and asked to answer the median $v$ of these three vertices. Let $x$ be the median of $r(v_1)$, $r(v_2)$ and $r(v_3)$. $x$ can be calculated in $O(\log n)$ time because each of $I[r(v_1),r(v_2)]$, $I[r(v_2),r(v_3)]$ and $I[r(v_3),r(v_1)]$ contains at most $9$ vertices and $x$ is the unique vertex in the intersection of these intervals. Now, we can state that $v\in F(x)$, because $F(x)$ is the only fiber that can intersect all of $I[v_1,v_2]$, $I[v_2,v_3]$ and $I[v_3,v_1]$.

Let $g_{v_1}$ (resp. $g_{v_2}$, $g_{v_3}$) be the gate of $v_1$ (resp. $v_2$, $v_3$) in $F(x)$, which can be calculated in $O(\log n)$ time. Then, from Lemma~\ref{lem:fiberintersection}, $v$ coincides with the median of $g_{v_1}$, $g_{v_2}$ and $g_{v_3}$. Therefore we can reduce the median query on the original graph into the median query on the fiber $F(x)$ in $O(\log n)$ time. By recursively working on the fiber, we can calculate $v$ after $O(\log n)$ recursion steps. Therefore the query can be answered in $O(\log^2 n)$ time in total. The data structure required here is constructed in $O(n\log^2 n)$ time just by taking the necessary parts of the algorithm in Appendix~\ref{sec:construction}.

\section{Acknowledgement}

We are grateful to our supervisor Prof. Hiroshi Hirai for supporting our work. He gave us a lot of ideas to improve our paper. In particular, he simplified the proofs and helped us improve the introduction and the overall structure of this paper. 
This work has been supported in part by The University of Tokyo Toyota-Dwango Scholarship for Advanced AI Talents.

\bibliographystyle{plainurl}
\bibliography{bib.bib}

\appendix
\newpage

\section{Finding the Gates, Entrances and Successors}\label{sec:finding}

In this section, we give a way to find the gates, entrances and successors, which is the remaining task in previous sections. 
The following lemma ensures that we can apply the discussion in Section~\ref{sec:single} and Section~\ref{sec:double} on the maximal tree with gated branches, not the total boundary. It is proved by a similar strategy as Lemma~7 in~\cite{chepoi2019distance}.

\begin{lemma}\label{lem:tiso}
$T$ is isometric.
\end{lemma}
\begin{proof}
Let $x,y\in V(T)$. Let $z$ be a median of $x,y$ and $r$. Let $z'$ be the lowest common ancestor of $x$ and $y$ in $T$. Denote the path on $T$ between $x$ (resp. $y$) and $r$ by $P_x$ (resp. $P_y$). Then, $z$ is on $P_x$ and $P_y$ and therefore $z$ is on $z'-r$ path on $T$. Here, $d(x,y)=d(x,z)+d(z,y)\geq d(x,z')+d(z',y)\geq d(x,y)$ hold. Therefore $z=z'$ and there is a $x-y$ shortest path that only uses the vertices of $T$.
\end{proof}

The following observation is useful.

\begin{lemma}\label{lem:treegate}
Let $T$ be a tree with gated branches rooted at $r$ and $T'$ be any convex subgraph of $T$ that contains $r$. Let $u\in V(G)$. Then, there is an imprint $w$ of $u$ in $T$ such that the gate of $u$ in $T'$ is the nearest ancestor of $w$ in $V(T')$ on $T$.
\end{lemma}
\begin{proof}
If $u$ has exactly one imprint, the statement is clear. Assume $u$ has two imprints $w^1$ and $w^2$ and let $t$ be the lowest common ancestor of $w^1$ and $w^2$. Let $z_1,z_2$ be the nearest ancestor of $w^1,w^2$ in $T'$, respectively. Then, either $z_1$ or $z_2$ is an ancestor of $t$ because otherwise we can apply Lemma~\ref{lem:quad} to obtain a vertex such that it is not $t$ and adjacent to two neighbors of $t$, which contradicts to the convexity of $T'$. Without loss of generality we can assume $z_2$ is an ancestor of $t$. If $z_1$ is also an ancestor of $t$, $z_1=z_2$ holds and the statement is clear. Assume that $z_1$ is a descendant of $t$. We prove $I[u,z_1]\cap V(T')=\{z_1\}$. Let $x\in V(T)\setminus \{z_1\}$. If $d(u,x)=d(u,w^1)+d(w^1,x)$ holds, $d(u,x)>d(u,z_1)$ holds because we have that $z_1$ is the nearest vertex from $w^1$ in $T'$ because of the isometricity of $T$. Otherwise, we have $d(u,x)=d(u,w^2)+d(w^2,x)<d(u,w^1)+d(w^1,x)$. In this case, we have that $t$ is on the $w^1-x$ path on $T$ and we have $d(u,z_1)<d(u,t)<d(u,t)+d(t,x)=d(u,x)$, which is from the isometricity of $T$. Therefore $x\not \in I[u,z_1]$ and the lemma is proved.
\end{proof}

The \emph{Euler-tour} of rooted tree $T$ with root $r$ is a walk $(r=w_0,\dots, w_{2n-1}=r)$ on $T$ that starts and ends at $r$ and passes through each edge exactly twice, in different direction. Given $T$, the Euler-tour of $T$ can be calculated in linear time by depth-first search.
Let $T'$ be a connected subgraph of rooted tree $T$. Let $S_{T'}$ be the set of the indices $i$ such that at least one of $(w_{i-1},w_{i})$ or $(w_{i},w_{i+1})$ is an edge in $T'$. If the Euler-tour of $T$ is already calculated, we can calculate such set $S_{T'}$ in $O(|V(T')|)$ time. Let $u\in V(T)$ and assume $u$ has an ancestor in $T'$. Let $u=w_i$. Then, the nearest ancestor of $u$ in $T'$ is the $w_{i'}$, where $i'$ is the largest index in $S_{T'}$ with $i'\leq i$. Such $i'$ can be found in $O(\log n)$ time by binary search if the elements of $S_{T'}$ are sorted. Therefore the following holds.

\begin{lemma}\label{lem:nearestds}
Let $T$ be a rooted tree and $T'$ be a connected subgraph of $T$.
Assume an Euler-tour of $T$ is given. Then, there is an algorithm such that, given a vertex $u\in V(T)$, calculate the nearest ancestor of $u$ in $T'$ in $O(\log |V(T')|)$ time. The preprocessing requires $O(|V(T')|\log |V(T')|)$ time and $O(|V(T')|)$ space.
\end{lemma}

Consider finding a successor. Let $T$ be a tree with gated branches and let $w,w'$ be two neighboring vertices in $T$ and $F(w),F(w')$ be the set of vertices with an imprint $w$ and $w'$, respectively. We consider finding $\su_{w,w'}(z)$, which is defined by the gate of $z$ in the boundary of $F(w')$ in Section~\ref{sec:process}.

For all possible pairs $(w,w')$, we precalculate the boundary $T_{w,w'}$ of $F(w)$ relative to $F(w')$. The size of it is bounded by $|F(w'')|$, where $w''$ is $w$ if $w$ is the child of $w'$ and $w'$ otherwise, and therefore the total size of these boundaries are $O(n)$. For each pair $(w,w')$, we construct the data structure in Lemma~\ref{lem:nearestds} on the total boundary of $F(w)$ that finds the nearest ancestor in $T_{w,w'}$. Then, from Lemma~\ref{lem:treegate}, we can calculate the gate. The same algorithm can also be applied to calculate the entrance under the setting of Section~\ref{sec:path}.

Now, we treat the remaining task in Section~\ref{sec:decomposition}. We use the same notation here as Section~\ref{sec:decomposition}; $m$ is the median of the cube-free median graph $G$, $F(x)$ is the fiber of $x$ in $\St(m)$, $r(u)$ is the vertex with $u\in F(r(u))$.
We give an algorithm that, given vertices $u,v\in V(G)$ and $x\in \St(m)$ with $I[u,v]\cap F(x)\neq \emptyset$, calculate the gate of $u$ in $F(x)$. Now, we can state that if $m\in I[r(u),x]$, the gate of $u$ in $F(x)$ is $x$; recall Lemma~\ref{lem:subgrid}. If $d(r(u),x)=1$, from the definition of the gate, the gate of $u$ in $F(x)$ is the neighbor of the gate of $u$ in the boundary of $F(r(u))$ relative to $F(x)$. Finally, assume $d(r(u),x)=2$. Let $y$ be the common neighbor of $r(u)$ and $x$. Then, again from the definition of the gate, we have the vertices $w_1,w_2,w_3$, where $w_1$ is the gate of $u$ in the boundary of $F(r(u))$ relative to $F(y)$, $w_2$ is the neighbor of $w_1$ in $F(y)$, $w_3$ is the gate of $w_2$ in the boundary of $F(x)$ relative to $F(y)$ and the gate of $u$ in $F(x)$ is the neighbor of $w_3$ in $F(x)$. Therefore, in all cases we can calculate the gate of $u$ in $F(x)$ by repeatedly calculating the gate, which can be computed in $O(\log n)$ time using Lemma~\ref{lem:treegate} and Lemma~\ref{lem:nearestds}.

Now we consider finding the entrance under the setting of Section~\ref{sec:single} and Section~\ref{sec:double}. Before doing it, we investigate the property of the tree with gated branches we actually treat.

\subsection{Basic Properties of a Maximal Tree with Gated Branches}\label{sec:gatedbranch}

Let $T$ be the maximal tree with gated branches rooted at $r$, here the tree with gated branches is maximal if we cannot add a vertex of $G$ to $T$ and get a tree with gated branches.
The following lemma characterizes the maximality of $T$.

\begin{lemma}\label{lem:maximal}
Let $T'$ be a tree with gated branches rooted at $r'$. Let $x$ be a vertex such that $x\not \in V(T')$. Then, $V(T')\cup \{x\}$ induces a tree with gated branches if and only if both of the following conditions holds.
\begin{description}
  \item[(a)] There is a vertex $y\in V(T)$ such that $x\sim y$ and $d(r',x)=d(r',y)+1$.
  \item[(b)] $r'=y$, or $x$ and $y$'s parent $z$ have no common neighbor other than $y$.
\end{description}
\end{lemma}
\begin{proof}
Let $T''$ be the subgraph induced by $V(T')\cup \{x\}$.
Assume $T''$ has gated branches. Since $x$ should be contained in some convex $r'-x$ path, there should be a neighbor $y$ of $x$ with $d(r',x)=d(r',y)+1$. Therefore (a) holds. If (b) does not hold, the common neighbor of $z$ and $x$ other than $y$ is contained in $I[x,z]$ and therefore $T''$ does not have gated branches. Therefore "only if" part is proved.

Now we prove "if" part. Assume both (a) and (b) hold. First, we prove $T''$ is a tree. Assume the contrary. Then, $x$ is contained in a cycle $C$ of $T''$ because $T'$ is a tree. We can take $C$ to contain $y$. $C$ has even length because the median graphs are bipartite. Let $w$ be the neighbor of $x$ in $C$ other than $y$. If $d(r',w)>d(r',y)$, $d(r',w)-2\geq d(r',y)$ holds and therefore $x\in I[r',w]$, which is a contradiction. If $d(r',w)<d(r',y)$, $d(r',w)\leq d(r',y)-2$ holds and therefore $x\in I[r',y]$, which is a contradiction. Therefore $d(r',w)=d(r',y)=d(r',x)-1$. Here we can apply Lemma~\ref{lem:quad} to $r',x,y,w$ and obtain a vertex $z$ with $d(r',z)=d(r',x)-2$, $z\sim y$ and $z\sim w$. Since $z\in I[r',y]$, $z$ is a parent of $y$. It contradicts to (b) and therefore $T''$ is a tree.

Now we prove $T''$ has gated branches. We should only to prove that the path $P$ between $x$ and $r'$ is convex. It is enough to prove the local convexity of $P$, and this is obtained from (b) and the local convexity of $V(P)\setminus \{x\}$.
\end{proof}

Now we consider identifying the entrance of $L'$ in the setting of Section~\ref{sec:single} and Section~\ref{sec:double}. First, we consider the former case.

\subsection{Entrance Identification: Single Imprint Case}\label{app:a}

Here we give the methods to identify the entrances of the staircases $L$ and $L'$ for the case that $I[u,v]$ contains exactly one imprint of $u$ in the maximal tree with gated branches $T$, which is the remaining problem in Section~\ref{sec:single}. We use the same settings and notations as Section~\ref{sec:single}; $w$ is the imprint of $u$ in $T$, $t$ is the lowest common ancestor of $w$ and $v$, $P$ (resp. $P'$) is any root-leaf path of $T$ that contains $w$ (resp. $v$), and $I[u,v]=I[u,t]\cup V(L')=I[u,w]\cup V(L)\cup V(L')$. We assume $t\neq v$ because otherwise $L'$ is empty. Let $e$ be the entrance of $L$. We assume $e\neq t$, otherwise $e$ is already found.

Let $t'$ be the neighbor of $t$ in $P'$ that is in $B_{L'}$; in other words, $t'$ is the start of the base of $L'$. Let $e'$ be the top of $L'$. We define the path $Q=(t=x_0,\dots, x_l=e)$ (resp. $Q'=(t'=x'_0,\dots, x'_l=e')$) as the $t-e$ (resp. $t'-e'$) shortest path.

The next lemma corresponds to the case of figure~(d). This lemma is the profit of imposing the maximality of $T$.

\begin{lemma}\label{lem:oneimp1}
Assume that $x_1$ is the neighbor of $t$ in $B_L$. Then, $e$ is on $P$.
\end{lemma}
\begin{proof}
Assume $e$ is not on $P$. Let $x_i$ be the vertex in $V(Q)\setminus V(B_L)$ with smallest index. From (i) in Lemma~\ref{lem:ppintersect}, $x_i\not\in V(T)$. We prove that $T\cup \{x_i\}$ still has gated branches by using Lemma~\ref{lem:maximal}. First, we have $d(r,x_{i-1})=d(r,t)+i-1=d(r,x_i)-1$. Because $x_1$ is on $B_L$, we have $i\geq 2$. From definition of $x_i$, we have that $x_{i-1}$ and $x_{i-2}$ is on $B_L$. There is no common neighbor of $x_i$ and $x_{i-2}$ other than $x_{i-1}$ because otherwise $Q$ is not locally convex. Hence the conditions in Lemma~\ref{lem:maximal} are satisfied and the lemma is proved. 
\end{proof}

Now, consider the algorithm that finds $e$. Let $F_t$ (resp. $F_{t'}$) be the set of the vertices with an imprint $t$ (resp. $t'$) in $T$. Let $T_{P',t}$ be the tree induced by the set of vertices in $T$ with gate $t$ in $P'$ and has a neighbor in $F_{t'}$.
Assume that $x_1$ is the neighbor of $t$ in $B_L$. Then, $e$ is the nearest ancestor of $w$ in $T_{P',t}$ in $T$. Thus, from Lemma~\ref{lem:nearestds}, $e$ can be calculated in $O(\log n)$ time with appropriate preprocessing. Since $|V(T_{P',t})|\leq |F_{t'}|$, the total size of $V(T_{P',t})$ is $O(n)$. Therefore we have following.

\begin{lemma}\label{lem:identify1}
There is an $O(\log n)$-time algorithm to identify $e$ in case $u$ has one imprint in $I[u,v]$ and $x_1$ is the neighbor of $t$ in $B_L$. The preprocessing requires $O(n)$ space and $O(n\log n)$ time.
\end{lemma}

Let us consider another case, which corresponds to the case of figure~(e). The next lemma holds.

\begin{lemma}\label{lem:oneimp2}
Assume $x_1$ is not on $B_L$. Then, $V(B_L)\cup V(B_{L'})$ is convex.
\end{lemma}
\begin{proof}
We prove that $V(B_L)\cup V(B_{L'})$ is locally convex. Let $t''$ be the neighbor of $t$ in $B_L$. Now, it is sufficient to prove that $t'$ and $t''$ have no common neighbor other than $t$. Assume there is such a vertex $z$. From the definition of the staircases, we have $z=x'_1$, which contradicts the fact that $x_1$ is not $t''$.
\end{proof}

From the results in Section~\ref{sec:path}, $V(L)\cup V(L')$ induces a staircases. Let $z$ be the endpoint of the path induced by the vertices of $L$ with an imprint $t$, other than $t$. Then, $e$ is the nearest ancestor of $z$ in the boundary of $F_t$ relative to $F_{t'}$. Thus, from Lemma~\ref{lem:nearestds}, $e$ can be calculated in $O(\log n)$ time with appropriate preprocessing.
Now, since the size of this boundary is bounded by $|F_{t'}|$, the total size of these boundaries is at most $O(n)$. Therefore we have the following.

\begin{lemma}\label{lem:identify2}
There is an $O(\log n)$-time algorithm to identify $e$ in case $u$ has one imprint in $I[u,v]$ and $x_1$ is not on $B_L$. The preprocessing requires $O(n)$ space and $O(n\log n)$ time.
\end{lemma}

Whether $x_1$ is on $B_L$ or not can be determined just by checking the neighbor of $t$ in $B_L$ has a neighbor with an imprint $t'$. If $e=t$, we can just apply the algorithm of the latter case to find $e$, before knowing $e=t$ holds. Therefore $e$ can be always calculated in $O(\log n)$ time with appropriate data structures.

\subsection{Entrance Identification: Double Imprints Case}\label{app:b}

Here we give the ways to identify the entrance of the staircases $L'$ for the case that $I[u,v]$ contains two imprints of $u$, which is the remaining problem in Section~\ref{sec:double}. We use the same settings and notations as Section~\ref{sec:double}; $w^1$ and $w^2$ are the imprints of $u$ in $T$, $w$ is the lowest common ancestor of $w^1$ and $w^2$, $t$ is the lowest common ancestor of $w$ and $v$, $P$ (resp. $P'$) is any root-leaf path that contains $w$ (resp. $v$), and $I[u,v]=I[u,t]\cup V(L')=I[u,w]\cup V(L)\cup V(L')$. First, we prove that $L$ is actually a line. Let $P_1$ (resp. $P_2$) be the root-leaf path of $T$ that contains $w^1$ (resp. $w^2$). We assume $t\neq w$ for a moment because otherwise $L$ is empty.

\begin{lemma}\label{lem:pathsqunion}
$I[u,t]$ is disjoint union of $I[u,w]$ and $B_{L}$.
\end{lemma}
\begin{proof}
Let $w'$ be the neighbor of $w$ in $P$ between $w$ and $t$. By definition of the staircases, it is sufficient to prove that $w'$ has no neighbor in $I[u,v]\setminus V(T)$. Assume the contrary and let $z$ be such a vertex. Then, there should be a common neighbor of $z$ and $w$. Since $w$ has already three neighbor in $I[u,v]\cap V(T)$, such common neighbor should be on $T$. Therefore we have that one of $P_1$ or $P_2$ is not convex and it is a contradiction.
\end{proof}

Now we identify the entrance $e$ of the staircases $L'$. We assume $t\neq v$ because otherwise $L'$ is empty. Let $t'$ be the neighbor of $t$ in $P'$ and $e'$ be the top of $L'$. Let $Q=(t=x_0,\dots, x_l=e)$ (resp. $Q'=(t'=x'_0,\dots, x'_l=e')$) be the $t-e$ (resp. $t'-e'$) shortest path. We prove the following lemma, which corresponds to the case of figure~(f).

\begin{lemma}\label{lem:twoimp2}
$e$ is on $T$.
\end{lemma}
\begin{proof}
First, we state that $x_1$ is on $T$. Indeed, if the contrary holds, $x_1$ should be in $I[u,w]$ because $x_1$ is not in either $L'$ or $L$, therefore $t=w$. However, from Lemma~\ref{lem:subgrid}, $w$ has at most two neighbors in $I[u,w]$, which already exist on $P_1$ and $P_2$, and there is no position remaining for $x_1$. Therefore $x_1$ is on $T$.

Now, assume $e$ is not on $T$. Let $x_i$ be the vertex in $V(Q)\setminus V(T)$ with smallest distance from $t$. Without loss of generality, we can assume that $x_{i-1}\in V(P_1)$. We prove that $T\cup \{x_i\}$ still has gated branches by using Lemma~\ref{lem:maximal}. First, we have $d(r,x_{i-1})=d(r,t)+i-1=d(r,x_i)-1$. From assumption, we have $i\geq 2$. By definition of $x_i$, we have $x_{i-1}$ and $x_{i-2}$ are on $P_1$. Now, there is no common neighbor of $x_i$ and $x_{i-2}$ other than $x_{i-1}$ because otherwise $Q$ is not convex. Hence the conditions in Lemma~\ref{lem:maximal} are satisfied and the lemma is proved.
\end{proof}

Now, consider the algorithm that find $e$. If we know whether $e$ is on $P_1$ or $P_2$, we can use the same algorithm as the Lemma~\ref{lem:identify1} to do it. To get the correct entrance, we just call the algorithm in Lemma~\ref{lem:identify1} on both $P_1$ and $P_2$ and return the one with larger distance from $t$, because these algorithms returns the nearest ancestor of $e$ on the corresponding path. Therefore we have the following.

\begin{lemma}\label{lem:identify3}
There is an $O(\log n)$-time algorithm to identify $e$ in case $u$ has two imprints in $I[u,v]$. The preprocessing requires $O(n)$ space and $O(n\log n)$ time.
\end{lemma}

\section{Construction}\label{sec:construction}

Here we prove that our data structure can be constructed in $O(n\log^3 n)$ time and $O(n\log^2 n)$ space. The whole precalculation algorithm is given in Algorithm~\ref{alg:preprocess}. Some of the parts are the same as the efficient construction of distance labeling scheme~\cite{chepoi2019distance}.

Let us see the details of Algorithm~\ref{alg:preprocess} one by one. Let $m$ be the median of $G$. For each $x\in \St(m)$, let $F(x)$ be the fiber of $x$ with respect to $\St(m)$, and let $T'_x$ be the maximal tree with gated branches that includes the total boundary of $F(x)$. For $x\in \St(m)$ and $w\in V(T'_x)$, let $F(x,w)$ be the set of vertices in $F(x)$ with an imprint $w$.

\begin{algorithm}
\caption{Construction of Data Structure}
\label{u_alg}
\begin{algorithmic}[1]\label{alg:preprocess}
\REQUIRE Cube-free median graph $G$
\STATE Calculate the median $m$ of $G$, using the algorithm in ~\cite{beneteau2019medians}\label{l1}
\STATE Calculate $d(m,v)$ for all $v\in V(G)$\label{l2}
\STATE Calculate $\St(m)$\label{l4}
\STATE Decompose $G$ into fibers with respect to $\St(m)$\label{l5}
\STATE Enumerate all squares in $G$ and store them appropriately\label{l3}
\FORALL{$x\in \St(m)$}
    \STATE Recursively construct a data structure on fiber $F(x)$\label{l6}
    \STATE Calculate maximal tree with gated branches $T'_x$ that includes the total boundary of $F(x)$\label{l7}
    \STATE Build a data structure that supports lowest common ancestor queries on $T'_x$\label{l8}
    \STATE Calculate the Euler-tour of $T'_x$\label{l8.5}
    \FORALL{$u\in V(F(x))$}
        \STATE Calculate the imprints of $u$ in $T'_x$\label{l9}
        \STATE For each imprint $w$ of $u$, calculate $p(I[u,w])$\label{l10}
        \STATE If $u$ has two imprints $w^1$ and $w^2$, let $w$ be the lowest common ancestor of $w^1$ and $w^2$ in $T'_x$ and calculate $p(I[u,w])$\label{l11}
    \ENDFOR
    \FORALL{$y\in \St(m)$ with $x\sim y$}
        \STATE Build a data structure that supports the queries to calculate the gate of $u$ in the boundary of $F(x)$ relative to $F(y)$ for all $u\in F(x)$\label{l11.2}
    \ENDFOR
    \FORALL{$w\in V(T'_x)$}
        \STATE Calculate total boundary of $F(x,w)$ and its Euler-tour\label{l11.5}
        \FORALL{$w'\in V(T'_x)$ with $w\sim w'$}
            \STATE Build a data structure that supports the queries to calculate the gate of $u$ in the boundary of $F(x,w)$ relative to $F(x,w')$ for all $u\in F(x,w)$\label{l11.8}
        \ENDFOR
        \STATE Do preprocessing in Lemma~\ref{lem:identify1}, Lemma~\ref{lem:identify2} and Lemma~\ref{lem:identify3}\label{l12}
    \ENDFOR
    \STATE Build a data structure presented in Section~\ref{sec:process}\label{l13}
\ENDFOR
\end{algorithmic}
\end{algorithm}

Line~\ref{l1} can be processed in linear time~\cite{beneteau2019medians}. Line~\ref{l2},~\ref{l4},~\ref{l5} can be implemented in linear time~\cite{chepoi2019distance}.

Let $xy_1zy_2$ be a square in $G$. Without loss of generality, we can assume $x$ has the largest distance from $m$. Then, $d(m,y_1)=d(m,y_2)=d(m,x)-1$ holds because $G$ is bipartite. Let us fix $x$. $x$ has no neighbor with distance $d(m,x)-1$ from $m$ other than $y_1$ and $y_2$, because otherwise we can obtain a cube by repeatedly applying Lemma~\ref{lem:quad}. Therefore we find $y_1$ and $y_2$ in constant time if we know $x$. By applying Lemma~\ref{lem:quad} on $m,x,y_1,y_2$, we obtain a vertex $z'$ with $y_1\sim z',y_2\sim z'$ and $d(m,z')=d(m,x)-2$. We have $z=z'$ because otherwise $G$ contains $K_{2,3}$. Since $y_1$ and $y_2$ has at most two neighbors with distance $d(m,x)-2$ from $m$, we can find $z$ in constant time if we know $y_1$ and $y_2$. Therefore, if we know $x$, we obtain the unique square $xy_1zy_2$ in constant time and we can enumerate all squares in linear time (Line~\ref{l3}). This also proves that there are at most $O(n)$ squares in $G$. We store the information of all squares in ascending order of the pairs of indices of vertices, for each diagonals in each square. It can be implemented in $O(n\log n)$ time (Line~\ref{l3}).


From Lemma~\ref{lem:half}, Line~\ref{l6} costs the time and space complexity by factor of $O(\log n)$. For $x\in \St(m)$, the total boundary can be calculated in linear time by definition. To calculate the maximal tree with gated branches, we first set the tree $T'_x$ as the total boundary of $F(x)$ and look through the vertices $z\in F(x)\setminus V(T'_x)$ in ascending order of distance from $x$. We check whether $z$ satisfies the conditions in Lemma~\ref{lem:maximal} one by one; if they are satisfied, we add $z$ to current $T'_x$. When we check $z$, we have only to check whether there exists a neighbor $z'$ of $z$ with $d(x,z')=d(x,z)-1$ and if exists, whether there exists a square that contains $z,z'$ and the parent of $z'$ (if the parent exists). We can check whether the square exists just by checking the unique square such that $z$ is the furthest vertex from $m$, which we already calculated in Line~\ref{l3}. Therefore Line~\ref{l7} can be processed in linear time. It is known that Line~\ref{l8} can be processed in linear time~\cite{bender2005lowest}. Line~\ref{l8.5} can be implemented in linear time.
Line~\ref{l9} can be processed in linear time by algorithm in~\cite{chepoi2019distance}.

Line~\ref{l10} can be implemented just by using the data structure constructed in Line~\ref{l6}. Line~\ref{l11} can also be similarly implemented along with a lowest common ancestor query. These lines take $O(n\log^2 n)$ time in total. 
Line~\ref{l11.2} takes $O(n\log n)$ time in total by Lemma~\ref{lem:nearestds}. Line~\ref{l11.5} takes linear time in total. Line~\ref{l11.8} also takes $O(n\log n)$ time in total by Lemma~\ref{lem:nearestds}. Line~\ref{l12} can be implemented in $O(n\log n)$ time, as we described in Section~\ref{sec:gatedbranch}. Finally, as we described in Appendix~\ref{sec:finding}, we can implement Line~\ref{l13} in $O(n\log n)$ time. Therefore, the total time complexity is $O(n\log^3 n)$.

In the algorithm of construction and query processing, we need to calculate the value $d(u,v)$ for $u,v\in V(G)$ efficiently. It can be done by using distance labeling scheme~\cite{chepoi2019distance}. We just need the values, for all $u\in V(G)$, of the distance between $u$ and each of the imprints, and it can be calculated when processing Line~\ref{l9}.

Now, we can construct the whole data structure we need in $O(n\log^3 n)$ time.

\section{Proof of Lemma~\ref{lem:convexstar}}\label{app:proof}

Here we give a proof of the following folklore result.

\begin{lemma}[{\rm Reappearing of Lemma~\ref{lem:convexstar}}]
Let $X$ be a convex vertex set and let $Y$ be a convex subset of $X$. For $x\in X$, let $F(x)$ be the fiber of $x$ with respect to $X$. Then, $\bigcup_{y\in Y}F(y)$ is convex.
\end{lemma}
\begin{proof}
Let $F=\bigcup_{y\in Y}F(y)$. By definition, $F$ induces a connected subgraph. We prove that $F$ is locally convex. 
Assume the contrary and let $y_1,y_2,y_3$ be vertices with $y_1,y_3\in F$, $y_2\not \in F$ and $y_1\sim y_2\sim y_3$.
Let $y_i\in F(x_i)$ for all $i=1,2,3$. Since two fibers $F(z)$ and $F(z')$ are neighboring if and only if $z\sim z'$, we have $x_1\sim x_2\sim x_3$. Now, we have $x_1\neq x_3$ because otherwise we have $y_2\in F(x_1)=F(x_3)$ from convexity of $F(x_1)$. Therefore, because the median graphs are bipartite, we have $d(x_1,x_3)=2$. Therefore $x_2\in I[x_1,x_3]\subseteq F$ holds and the lemma is proved.
\end{proof}

\begin{algorithm}
\caption{The Algorithm to Interval Query with One End on $T$}
\begin{algorithmic}[1]\label{alg:stairsall}
\REQUIRE $u,v\in V(G)$
\IF{$d(u,v)=d(u,w)+d(w,v)$ holds for exactly one imprint $w$ of $u$}
    \STATE Let $t$ be the lowest common ancestor of $w$ and $v$
    \STATE Let $P=(r=w_0,\dots, w_k)$ and $P'=(r=w'_0,\dots, w'_k)$, where $w_a=w'_a=t$ and $w_b=w$ be the root-leaf path that contains $w$ and $v$, respectively
    \IF{$w\neq t$}
        \STATE {\it(Consider the decomposition of $I[u,v]$ into $I[u,w]$, $L=L(s,w_{b-1},t)$ and $L'=L(s',w'_{a+1},v)$)}
        \STATE Find the entrance $e$ of $L$
        \STATE $(z,S)\leftarrow \text{StaircasesQuery}(s,w_{b-1},t)$, where $s$ is the neighbor of $e$ in $F(w_{b-1})$
        \IF{$v=t$}
            \RETURN $p(I[u,w])\oplus S$
        \ENDIF
        \IF{$w_{a+1}$ has a neighbor in $F(w'_{a+1})$}
            \STATE Find the entrance $e'$ of $L'$ using the data structure in Lemma~\ref{lem:identify1}
        \ELSE
            \STATE Find the entrance $e'$ of $L'$ using the data structure in Lemma~\ref{lem:identify2}
        \ENDIF
        \STATE $(z',S')\leftarrow \text{StaircasesQuery}(s',w'_{a+1},v)$, where $s'$ is the neighbor of $e'$ in $F(w'_{a+1})$
        \RETURN $p(I[u,w])\oplus S\oplus S'$
    \ELSE
        \IF{$v=t$}
            \RETURN $p(I[u,w])$
        \ENDIF
        \STATE {\it(Consider the decomposition of $I[u,v]$ into $I[u,w]$ and $L'=L(s',w'_{a+1},v)$)}
        \STATE Find the entrance $e'$ of $L'$
        \STATE $(z',S')\leftarrow \text{StaircasesQuery}(s',w'_{a+1},v)$, where $s'$ is the neighbor of $e'$ in $F(w'_{a+1})$
        \RETURN $p(I[u,w])\oplus S'$
    \ENDIF
\ELSE
    \STATE Let $w^1,w^2$ be the imprints of $u$
    \STATE Let $w$ be the lowest common ancestor of $w^1$ and $w^2$
    \STATE Let $t$ be the lowest common ancestor of $w$ and $v$
    \STATE Let $P_1=(r=w_{1,0},\dots,w_{1,k_1})$, $P_2=(r=w_{2,0},\dots, w_{2,k_2})$ and $P'=(r=w'_0,\dots, w'_{k'})$, where $w_{1,a}=w_{2,a}=w'_{a}=t$, $w_{1,b}=w_{2,b}=w$ be the root-leaf path that contains $w^1$, $w^2$ and $v$, respectively
    \STATE {\it(Consider the decomposition of $I[u,v]$ into $I[u,w]$, the path $(w_{1,a},\dots, w_{1,b-1})$ and $L'=L(s',w'_{a+1},v)$)}
    \STATE Find the entrance $e'$ of $L'$ using the data structure in Lemma~\ref{lem:identify3}
    \STATE $(z',S')\leftarrow \text{StaircasesQuery}(s',w'_{a+1},v)$, where $s'$ is the neighbor of $e'$ in $F(w'_{a+1})$
    \IF{$w=t$}
        \RETURN $p(I[u,w])\oplus S'$
    \ENDIF
    \STATE $(z,S)\leftarrow \text{StaircasesQuery}(w_{1,b-1},w_{1,b-1},w_{1,a})$
    \RETURN $p(I[u,w])\oplus S\oplus S'$
\ENDIF
\end{algorithmic}
\end{algorithm}

\end{document}